%% file: main.tex
\documentclass[a4,11pt]{article}
\usepackage{a4wide}
\usepackage{tikz} 
\usetikzlibrary{positioning,chains,fit,shapes,calc,arrows}
\usepackage{sidecap}
\usepackage{tikz-cd}
\usepackage{amsthm}
\usepackage{amsmath}
\usepackage{amssymb}
\usepackage{mathtools}
\usepackage{algorithm}
\usepackage[noend]{algpseudocode}
\usepackage{graphicx}
\usepackage{caption}
\usepackage{subcaption}
\usepackage{comment}

\newcommand{\abs}[1]{| #1 |}

\newcommand{\Vertices}{
            \begin{scope}[start chain=going below,node distance=7mm]
            \foreach \i in {1,2}
              \node[on chain] (a\i) {$a_\i$};
            \end{scope}
            
            \begin{scope}[xshift=4cm,yshift=0.5cm,start chain=going below,node distance=7mm]
            \foreach \i in {1,2,3}
              \node[on chain] (g\i) {$g_\i$};
            \end{scope}
            }

\newtheorem{lemma}{Lemma}
\newtheorem{theorem}{Theorem}
\newtheorem{observation}[theorem]{Observation}
\newtheorem{corollary}{Corollary}
\newtheorem{definition}{Definition}
\newtheorem{example}{Example}

\newcommand{\KM}[1]{\textcolor{blue}{KM: #1}}
\newcommand{\BRC}[1]{\textcolor{red}{BRC: #1}}

\newcommand{\HA}[1]{\textcolor{blue}{Hana: #1}}

\newcommand{\OPT}{\mathit{OPT}}

\DeclarePairedDelimiter\ceil{\lceil}{\rceil}
\DeclarePairedDelimiter\floor{\lfloor}{\rfloor}

\newcommand{\OPTH}{\OPT^H}
\newcommand{\EH}{E^H}
\newcommand{\EL}{E^L}
\newcommand{\AH}{A^H}
\newcommand{\BH}{B^H}
\newcommand{\CH}{C^H}
\newcommand{\NSW}{\mathrm{NSW}}
\newcommand{\APXAlg}{\mathit{APXAlg}}
\newcommand{\uv}{u}
\usepackage{authblk}
\usepackage{hyperref}

\title{Nash Social Welfare for 2-value Instances}
\author[1]{Hannaneh Akrami}
\author[2]{Bhaskar Ray Chaudhury}
\author[3]{Kurt Mehlhorn}
\author[4]{Golnoosh Shahkarami}
\author[5]{Quentin Vermande}
\affil[1]{Max Planck Institute for Informatics, Universität des Saarlandes, \href{mailto:hakrami@mpi-inf.mpg.de}{hakrami@mpi-inf.mpg.de}}
\affil[2]{Max Planck Institute for Informatics, Universität des Saarlandes, \href{mailto:braycha@mpi-inf.mpg.de}{braycha@mpi-inf.mpg.de}}
\affil[3]{Max Planck Institute for Informatics, \href{mailto:mehlhorn@mpi-inf.mpg.de}{mehlhorn@mpi-inf.mpg.de}}
\affil[4]{Max Planck Institute for Informatics, Universität des Saarlandes, \href{mailto:gshahkar@mpi-inf.mpg.de}{gshahkar@mpi-inf.mpg.de}}
\affil[5]{\'{E}cole Normale Sup\'{e}rieure, Paris, \href{mailto:quentin.vermande@ens.fr}{quentin.vermande@ens.fr}}

\begin{document}
\maketitle

\input{Integer}

\bibliography{main}
\bibliographystyle{abbrv}

\end{document}

%% file: Integer.tex
\begin{abstract}
	This paper is merged with \cite{fullversion}. We refer the reader to the full and updated version.
	
    We study the problem of allocating a set of indivisible goods among agents with 2-value additive valuations. Our goal is to find an allocation with  maximum Nash social welfare, i.e.,  the geometric mean of the  valuations of the agents. We give a polynomial-time algorithm to find a Nash social welfare maximizing allocation when the valuation functions are \emph{integrally 2-valued}, i.e., each agent has a value either $1$ or $p$ for each good, for some positive integer $p$. We then extend our algorithm to find a better approximation factor for general 2-value instances.     
\end{abstract}

\section{Introduction}
Fair division of goods has developed into a fundamental field in economics and computer science. In a classical fair division problem, the goal is to allocate a set of goods among a set of agents in a \emph{fair} (making every agent content with her bundle) and \emph{efficient} (achieving good overall welfare) manner. One of the most well studied class of valuation functions are \emph{additive} valuation functions, where the utility of a bundle is the sum of utilities of the individual goods in the bundle. When agents have {additive} valuation functions,  the Nash social welfare or equivalently the geometric mean of the valuations, $\Big(\prod_{i \in [n]} v_i(X_i) \Big)^{1/n}$, is a direct indicator of the fairness and efficiency of an allocation. In particular, any allocation that maximizes Nash social welfare is \emph{envy-free up to one good (EF1)}, i.e.,  no agent envies another agent following the removal of \emph{some} single good from the other agent's bundle and \emph{Pareto-optimal}, i.e.,  no allocation can give a single agent a better bundle without giving a worse bundle some other agent~\cite{CaragiannisKMP016}.  Unfortunately, maximizing Nash social welfare is APX-hard~\cite{Lee17} and allocations that achieve good approximations of Nash social welfare may not have similar fairness and efficiency guarantees. Despite this, finding good approximations of Nash social welfare has received substantial  interest over the years~\cite{ColeG18, BKV18, AnariGSS17}.

Although Nash social welfare maximization is hard when agents have general additive valuation functions, some special cases are polynomial time solvable. One of the interesting special cases is when agents have \emph{binary additive valuations}, i.e., when for each agent $i$ and each good $g$, we have $v_i(\{g\}) \in \{0,1\}$. Although, this class of valuation functions seem restrictive in their expressiveness of individual preferences, several real life scenarios involve preferences that are dichotomous and as a result there is substantial research on fair division under binary valuations~\cite{binary1, barman2018binarynsw, BL16, DarmannS15, binary3, HalpernPPS20}.  Barman et al.\ \cite{barman2018binarynsw} give a  polynomial time algorithm to find an allocation with maximum Nash social welfare when agents have binary additive valuation functions. Furthermore, for  binary valuations, Halpern et al.\ \cite{HalpernPPS20} show that determining a fair allocation via Nash social welfare maximization is also \emph{strategyproof}, i.e., agents do not benefit by misreporting their preferences. A generalization of binary valuation functions are 2-value functions, where for each agent $i$ and each good $g$, we have $v_i(\{g\}) \in \{a,b\}$, for some $a,b \geq 0$\footnote{2-value functions are binary valuation functions when $a=0$ and $b=1$. Additionally, for all non-binary 2-value instances, one can assume without loss of generality that $a=1$ and $b >a $}. Amanatadis et al.\ \cite{AmanatidisBFHV21} show that even when agents have 2-value functions, an allocation with maximum Nash social welfare implies stronger fairness notions such as envy-freeness up to any good (EFX), where no agent envies another agent following the removal of \emph{any} single good from the other agent's bundle. Thus, even for the more general 2-value instances, finding an allocation with maximum Nash social welfare is a canonical way of dividing goods fairly and efficiently. However, Amanatadis et al.\ \cite{AmanatidisBFHV21}, leave the problem of maximizing Nash social welfare for 2-value instances open. 

In this paper, we take steps towards solving this open problem. We consider the problem of maximizing Nash social welfare when the valuation functions of the agents are \emph{integrally 2-value}, i.e., for each agent $i$ and for each good $g$, we have $v_i(\{g\}) \in \{1,p\}$, for some positive integer $p$. The main result of our paper is a polynomial time algorithm for maximizing Nash social welfare when the valuation functions of the agents are integrally 2-value. 

\newtheorem*{IntegerResult}{Theorem \ref{IntegerResult}}
\begin{IntegerResult}
There exists a polynomial-time algorithm for maximizing Nash social welfare when the valuation functions of the agents are integrally 2-value.
\end{IntegerResult}

We remark that an immediate corollary of Theorem~\ref{IntegerResult} is a better approximation of Nash social welfare when agents have 2-value functions (when $p$ may not be integral): For any instance where $p$ is not integral, we can round $p$ to the closest integer ($\ceil{p}$ or $\floor{p}$) and run our algorithm for integral $p$. It is not hard to see that this achieves a $\max \{ \frac{\floor{p}}{p}, \frac{p}{\ceil{p}}\}-$approximation of the maximum Nash social welfare. Also note that  $\max \{ \frac{\floor{p}}{p}, \frac{p}{\ceil{p}}\} \geq \sqrt{2}$ which is better than the best approximation of $e^{1/e}$ known for general additive valuations~\cite{BKV18}. 

We now highlight our main technical contributions.


\subsection{Our Techniques}

In this section, we give a brief overview of the main ideas and techniques used by our algorithm. 

We define the problem of maximizing Nash social welfare as a graph problem: We have a weighted complete bipartite graph  with the set of agents and the set of goods being the independent sets. The edge-weights represent the value of a good for an agent and are either $1$, i.e., \emph{light edge}, or $p$, i.e., \emph{heavy edge}. We say that a good is heavy if it has at least one incident heavy edge and light otherwise.  An allocation is a multi-matching in which all goods have degree at most one and agents can have degrees larger than one. This way of defining allocation allows us to use the idea of the augmenting paths, like the algorithm in~\cite{barman2018binarynsw}.

We start by mentioning some crucial structural differences of 2-value instances to {binary instances}. For binary instances, Halpern et al.\ \cite{HalpernPPS20} show that an allocation maximizes Nash social welfare if and only if it is \emph{leximax}\footnote{In \cite{HalpernPPS20} it is called leximin. However we find leximax more expressive.}, i.e., the utility profile of the allocation is lexicographically maximum. However, this is not true for 2-value instances. Consider the following example: There are two agents $a_1$, $a_2$ and five goods $g_1$, $g_2$, $g_3$, $g_4$, and $g_5$ and $p=5$. All goods are light for $a_2$ and agent $a_1$ values $g_1$ and $g_2$ heavily, and the other goods light. One can verify that a Nash social welfare maximizing allocation is the one where $a_1$ gets $\{ g_1, g_2 \}$ and $a_2$ gets $\{ g_3, g_4, g_5 \}$. However, this allocation is not leximax, as it is lexicographically dominated by the allocation $a_1 \gets \{ g_1 \}$, $a_2 \gets \{g_2, g_3, g_4, g_5 \}$. This necessitates finding some  \emph{other tractable characterization} of a Nash social welfare maximizing allocation, in particular a characterization of the allocation of heavy goods. This motivates the main structure of our algorithm: allocate the heavy goods carefully and then allocate the light goods greedily.

\paragraph{Characterizing the allocation of heavy goods.} The main bulk of our effort is in finding the correct allocation of the heavy goods. We briefly elaborate our technique that achieves this goal. Firstly, we give a nice characterization of the heavy goods allocation of a Nash social welfare maximizing allocation. We refer to the \emph{heavy-part} $A^H$ of an allocation $A$ as the set of all heavy edges in the allocation, and we call an allocation a \emph{heavy-only} allocation if the allocation contains only heavy edges, i.e., if $A^H = A$. One of our main structural results (shown in Theorem~\ref{heavysim}) is that there exists a Nash social welfare maximizing allocation $\OPT$, such that the heavy-part of $\OPT$ is leximax among all heavy-only allocations of the same cardinality. Therefore, if we know the number of heavy-edges in $\OPT$, then the utility profile of the heavy-part of $\OPT$ is unique (as it is leximax).  Thus, the main question boils down to finding a heavy-only allocation, which is leximax among all heavy-only allocations of the same cardinality, and has equal number of heavy-edges as that in $\OPT$\footnote{At this point, we make a subtle but important clarification. Note that a Nash social welfare maximizing allocation need not allocate all heavy goods along heavy edges. Consider a simple scenario where there are two agents $a_1$ and $a_2$, and there are two goods $g_1$ and $g_2$. Both $g_1$ and $g_2$ are heavy to $a_1$ and both are light to $a_2$. Both $g_1$ and $g_2$ are heavy goods. However, in an optimal allocation, one of the heavy goods is not allocated to an agent who finds it heavy. Thus, it is not immediate how to find the number of heavy edges in $\OPT$.}. 

\paragraph{Finding the right allocation of heavy goods.} The crucial technical barrier lies in the fact that we do not know the number of heavy edges in $\OPT$.  We briefly elaborate how we overcome this connundrum. We start the algorithm by finding a heavy-only allocation $A$ that maximizes Nash social welfare, is leximax  and subject to this,  has the highest number of heavy edges (such an allocation can be determined by adapting the algorithm of Barman et al.\ \cite{barman2018binarynsw}). Then, we allocate the light goods greedily to $A$, i.e., we iterate through the unallocated light goods and allocate a light good to an agent with smallest utility.  Note that by definition of $A$, the total number of heavy edges in $A$ is larger than or equal to that in $\OPT$, i.e., $|A^H| \geq |\OPT^H|$. Thereafter, as our second main result, we show that the allocation $A$ (after allocation of the light goods), exhibits \emph{local sensitivity to the heavy edges}, i.e., if the number of heavy edges is larger than that in $\OPT$, then a simple local reallocation can improve the Nash social welfare. In particular, given the allocation $A$, where $A^H$ is leximax among all heavy-only allocations of the same cardinality, if the number of heavy edges in $\OPT$ is smaller than that in $A$, then we can find another allocation $\hat{A}$ from $A$, by moving a heavy good from an agent with highest utility to an agent with lowest utility. Furthermore, we can guarantee that,
\begin{itemize}
    \item $\hat{A}$ is leximax among all heavy-only allocations of the same cardinality, 
    \item $|\hat{A}^H| = |A^H|-1$, and 
    \item the Nash social welfare of $\hat{A}$ is at least the Nash social welfare of $A$.
\end{itemize}
We 
explain how this property helps us circumvent the issue of not knowing the number of heavy edges in $\OPT$. Our algorithm starts with allocation $A$. We move a heavy good from an agent with highest utility to an agent with lowest utility as long as the Nash social welfare of the allocation improves. Let our final allocation be $\tilde{A}$. Note that $|\tilde{A}^H| \leq |\OPT^H|$, as otherwise we can still improve the Nash social welfare by the aforementioned reallocation. Also note that every time we perform the reallocation, the cardinality of the heavy part of the allocation decreases by exactly one. Since $|A^H| \geq |\OPT^H|$ and $|\tilde{A}^H| \leq |\OPT^H|$, our algorithm must have constructed an allocation $\hat{A}$ during the transition from $A$ to $\tilde{A}$, such that $|\hat{A}^H| = |\OPT^H|$. Thus, the heavy part of $\hat{A}$  is leximax among all heavy-only allocations of the same cardinality and has the number of heavy edges equal to that of $\OPT$. Therefore, $\hat{A}$ is our desired allocation. However, since we have that Nash social welfare of $\tilde{A}$ is at least the Nash social welfare of $\hat{A}$. $\tilde{A}$ is also a Nash social welfare maximizing allocation.


\subsection{Further Related Work}
 There are several polynomial time algorithms that find allocations achieving an $\mathcal{O}(1)$ approximation of the maximum Nash social welfare~\cite{ColeG18, BKV18, AnariGSS17}. The algorithm by Barman et al.\ \cite{BKV18} also achieve additional properties of fairness and efficiency like approximate EF1 and approximate Pareto-optimality. The Nash social welfare maximization has $\mathcal{O}(1)$ approximation algorithms, even when agents have more general valuation functions than additive valuation functions \cite{GargHM18, AnariMGV18, ChaudhuryCG0HM18}.  When agents have submodular and subadditive valuations, algorithms with approximation factors (almost) linear in $n$  had been obtained \cite{GargKK20, CGMehta21, BBKS20}. The $\mathcal{O}(n)$-approximation is also best approximation one can achieve with polynomially many value queries\footnote{In a value query, given an agent $i$ and a set $S$, the output is $v_i(S)$ where $v_i(\cdot)$ is the valuation function of agent $i$.} when agents have subadditive valuations~\cite{BBKS20}. Very recently, Li and Vondr{\'{a}}k ~\cite{LiVondrak21} improved the approximation factor from $\mathcal{O}(n)$ to $\mathcal{O}(1)$ when agents have submodular valuations. 
 
 There is also literature in guaranteeing high Nash social welfare with other fairness notions. For instance, relaxations of EFX can be guaranteed with high Nash welfare~\cite{CaragiannisGravin19, CGMehta21}, approximations of \emph{groupwise maximin share} (GMMS)~\cite{CKMS20} and \emph{maximin share} (MMS)~\cite{CKMS20, CaragiannisKMP016} are achieved with high Nash welfare.

\subsection{Independent Work}
In private communication, we are aware that similar results are obtained by Jugal Garg and Aniket Murhekar~\cite{GM21}. They also obtain a polynomial time algorithm to determine an allocation with maximum Nash welfare for instances where the valuation functions of the agents are integrally 2-valued.

\section{Preliminaries}

	We have a set $N$ of $n$ agents and a set $M$ of $m$ goods. Each agent $i$ has a utility $\uv_i$. Utilities are 2-value additive, i.e., $\uv_i(S) = \sum_{s \in S} \uv_i(\{s\})$ where $\uv_i(\{s\}) \in \{1,p\}$ for each $i \in N$ and $s \in M$ and $p$ is an integer greater than $1$.

	We maximize $\NSW$ which is the geometric mean of the utilities of agents for their bundles. Formally, for an allocation $A$ which assigns the bundle $A_i$ to agent $i$, $\NSW(A) = (\Pi_{i=1}^{n} \uv_i(A_i))^{1/n}$. The goal is to find an allocation maximizing $\NSW$. This notion defined by Nash \cite{Nash50} in 1950s captures two important properties of a desired allocation; efficiency and fairness. By $\NSW(X, \uv)$ we mean the Nash social welfare of allocation $X$ under utility vector $\uv$. In case $\uv$ is clear from the context, we might drop it and use $\NSW(X)$.

    \subsection{Utility Graphs}
    
	In most of the papers working on fair division, an allocation is defined as an $n$-tuple of disjoint bundles allocated to agents; i.e $A = (A_1, A_2, \dots, A_n )$ where $A_i$ is allocated to agent $i$. According to our techniques which heavily employ ideas similar to augmentation in matching algorithms, we find it more convenient to define an allocation from a graph point of view. Consider the complete bipartite graph $G = (N \cup M, E)$ where we have agents on one side and goods on the other side. We call the edge between agent $i$ and good $g$ \emph{heavy}, if $\uv_i({g}) = p$ and \emph{light} otherwise. We use $\EH$ and $\EL$ to denote the set of \emph{heavy} and \emph{light} edges respectively.  Moreover, good $g$ is \emph{heavy} for agent $i$ if $\uv_i(g) = p$ and it is \emph{light} for her otherwise. Figure \ref{fig1} shows an instance with $2$ agents and $3$ goods.
	
	An \emph{allocation} is a subset $A$ of $E$ such that for each $g \in M$ there is at most one edge in $A$ incident to $g$. Note that allocations are partial. If there is an edge $(i,g) \in A$, we say that $g$ is assigned to $i$ in $A$ or $i$ owns $g$ in $A$ or $A$ assigns $g$ to $i$. Otherwise, $g$ is unassigned. An allocation is \emph{complete} if all goods are assigned. For an agent $i$, we use $A_i$ for the set of goods assigned to $i$ in $A$. We refer to $A_i$ as the bundle of $i$ in $A$. Then $\uv_i(A_i)$ is the utility of $i$'s bundle for $i$. Figure \ref{fig2} shows an allocation for the instance shown in Figure \ref{fig1}.
	
	The \emph{utility vector} of an allocation $A$ is the vector $(\uv_1(A_1 ),$ $\ldots, \uv_n(A_n))$ and its \emph{utility profile} is the utility vector sorted in the non-descending order of utilities. A utility profile $(b_1,\ldots,b_n)$ is \emph{lexicographically larger} than a utility profile $(c_1,\ldots,c_n)$ or $(b_1,\ldots,b_n)\succ_{\mathit{lex}}(c_1,\ldots,c_n)$ if the profiles are different and $b_i > c_i$ for the smallest $i$ with $b_i \not= c_i$. An allocation $A$ with utility profile $(a_1,\ldots,a_n)$ is \emph{leximax} in a family $\cal A$ of allocations if for no allocation $B \in \cal A$ with utility profile $(b_1,\ldots,b_n)$, $(b_1,\ldots,b_n)\succ_{\mathit{lex}} (a_1,\ldots,a_n)$.  
	
	\begin{definition}
	   (\emph{Heavy-only allocation}) For an allocation $A$, its \emph{heavy part} $\AH$ is the restriction of $A$ to the heavy edges, i.e., $\AH = A \cap \EH$. An allocation is \emph{heavy-only} if $A = \AH$. Alternatively, $A \subseteq \EH$. 
	\end{definition}

	For an agent $i$, $\AH_i$ is the set of heavy edges incident to agent $i$ under allocation $A$.
	We refer to $\abs{\AH_i}$ as the \emph{heavy degree} of $i$ in $A$ and denote it $\deg_H(i,A)$ or $\deg_H(i)$. 
	\begin{figure}
        \centering
        \begin{subfigure}{7cm}
            \centering
            \begin{tikzpicture}[every node/.style={draw,circle}]]
            
            \Vertices
            
            \foreach \i in {1,2}
                \foreach \j in {1,2,3}
                    \draw[-, red] (a\i) -- (g\j);
                    
            \draw[-, ultra thick] (a1) -- (g1);
            \draw[-, ultra thick] (a1) -- (g2);
            \draw[-, ultra thick] (a2) -- (g2);
            \draw[-, ultra thick] (a2) -- (g3);
           
        \end{tikzpicture}
        \subcaption{}
	    \label{fig1}
        \end{subfigure}
        \begin{subfigure}{7cm}
            \centering
            \begin{tikzpicture}[every node/.style={draw,circle}]
            
            \Vertices
            
            \draw[-, ultra thick] (a1) -- (g1);
            \draw[-, ultra thick] (a2) -- (g2);
            \draw[-, red] (a1) -- (g3);
           
        \end{tikzpicture}

	    \subcaption{}
	    \label{fig2}
	   \end{subfigure}
	   \caption{The graph $G$ in (a) corresponds to an instance with $2$ agents on the left side and $3$ goods on the right. Thick black edges and thin red edges correspond to heavy and light edges, respectively.
	   The graph $G$ in (b) corresponds to a complete allocation $A$ in which $g_1$ and $g_3$ are allocated to $a_1$ and $g_2$ is allocated to $a_2$. Note that this allocation does not maximize $\NSW$. $G$ restricted to black edges is $\AH$. We have $deg_H(a_1)=deg_H(a_2)=1$.}
    \end{figure}
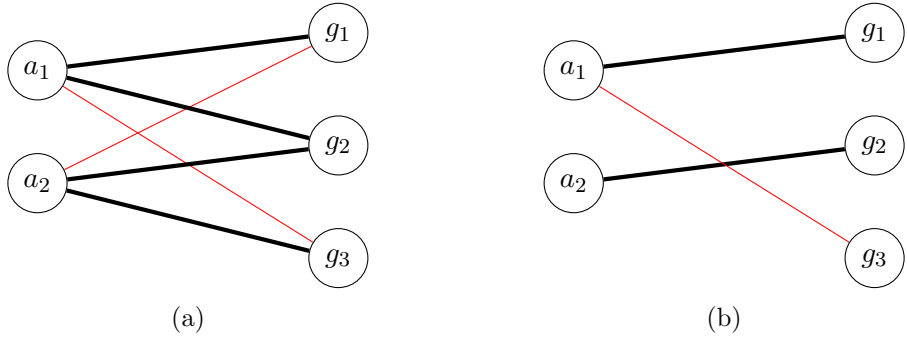

	\subsection{Alternating Paths}\label{altpath}
	In our definition of allocations, they correspond to multi-matchings. Later, in order to improve one allocation, we alternate the edges along a path consisting of every other edge inside the current multi-matching. Hence, it is useful to define alternating paths.
	
	\begin{definition}
	    (\emph{Heavy alternating path}) An \emph{alternating path} with respect to an allocation $A$ is any path whose edges are alternating between $A$ and $E \setminus A$. A \emph{heavy alternating path} is an alternating path all of whose edges belong to $\EH$. 
	\end{definition}
	See Figure \ref{fig3} for an example of a heavy alternating path.
    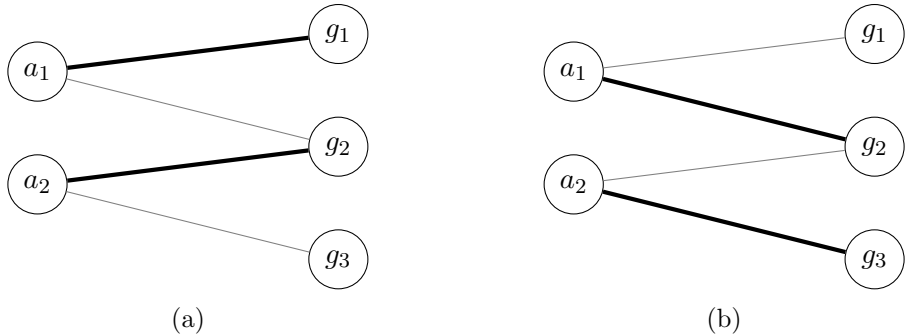
\begin{figure}
        \centering
        \begin{subfigure}{7cm}
            \centering
            \begin{tikzpicture}[every node/.style={draw,circle}]
            
            \Vertices
            
            \draw[-, ultra thick] (a1) -- (g1);
            \draw[-, ultra thick] (a2) -- (g2);
            \draw[-, gray] (a1) -- (g2);
            \draw[-, gray] (a2) -- (g3);
            
            \end{tikzpicture}
            \subcaption{}
        \end{subfigure}
        \begin{subfigure}{7cm}
            \centering
            \begin{tikzpicture}[every node/.style={draw,circle}]
            
            \Vertices
            
            \draw[-, gray] (a1) -- (g1);
            \draw[-, gray] (a2) -- (g2);
            \draw[-, ultra thick] (a1) -- (g2);
            \draw[-, ultra thick] (a2) -- (g3);
            
        \end{tikzpicture}

	    \subcaption{}
        \end{subfigure}
        \caption{The path $P=(g_1, a_1, g_2, a_2, g_3)$ in (a) is a heavy alternating path where all the edges are heavy and the thick black edges shows allocation $A$.
        (b) shows $A \oplus P$.}
        \label{fig3}
    \end{figure}

	\begin{definition}
	    (\emph{Alternating path wrt two allocations})
	    An \emph{alternating path} with respect to  two allocations $A$ and $B$ is any path whose edges are alternating between $A \setminus B$ and $B \setminus A$, i.e., between edges only in $A$ and edges only in $B$. 
	\end{definition}
	
	An \emph{alternating path decomposition} is defined with respect to two heavy-only allocations $A$ and $B$. The graph $A \oplus B$ is defined on the same set of vertices as in $A$ and $B$. Moreover, the edge $e$ appears in $A \oplus B$, if and only if $e$ is in exactly one of $A$ or $B$.
	We want to decompose $A \oplus B$ into edge-disjoint paths.
	Note that in $A \oplus B$, goods have degree zero, one, or two. 
	For a good of degree two, the two incident edges belong to the same path. For an agent $i$, let $a_i$ ($b_i$) be the number of $A$-($B$)-edges incident to $i$ in $A \oplus B$. Then we have $\min(a_i,b_i)$ alternating paths passing through $i$,  $\max(0, a_i - b_i)$ alternating paths starting in $i$ with an edge in $A$, and $\max(0, b_i - a_i)$ alternating paths starting in $i$ with an edge in $B$.

    If $P$ is an even length heavy alternating path with respect to $A$ connecting two agents $i$ and $j$ with the edge of $P$ incident to $i$ in $\EH \setminus A$ and the edge incident to $j$ in $\AH$, then 
	$A \oplus P$ contains the same number of heavy edges as $A$, i.e., $\abs{\AH} = \abs{(A \oplus P)^H} = \abs{\AH \oplus P}$. Moreover, the heavy degree of $i$ increased, the heavy degree of $j$ decreased and all other heavy degrees are unchanged.  
	
	\begin{example}
	    Consider the example shown in Figure \ref{fig1} and allocation $A$ with 
	    $$A^H = \{(a_1,g_1), (a_2, g_2)\}$$
	    shown in Figure \ref{fig2}. Let $B$ be another allocation for which 
	    $$B^H = \{(a_1,g_1),(a_1,g_2)\}.$$
	    Then, Figure \ref{figPathDecompose} shows $A^H \oplus B^H$. Black edges are only in $A^H$ and green edges are only in $B^H$. The path decomposition of $A \oplus B$ is the only path $P=(a_1,g_2,a_1)$ that exists in this graph.
	\end{example}
	
	\begin{figure}
	\centering
	    \begin{tikzpicture}[every node/.style={draw,circle}]
            
            \Vertices
            
            (a1) edge [bend right] node (g1)
            \draw[-, ultra thick, green] (a1) -- (g2);
            \draw[-, ultra thick] (a2) -- (g2);
            
        \end{tikzpicture}

	    \caption{$A^H \oplus B^H$}
	    \label{figPathDecompose}
	\end{figure}
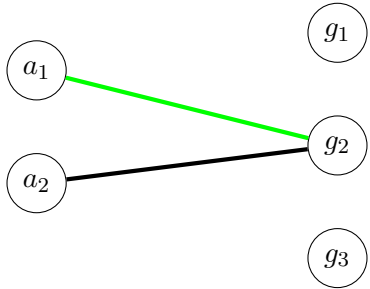	
	We use the following notions in the rest of the paper. 
	\begin{align*}
	    \OPT &= \text{an allocation maximizing $\NSW$}\\
	    \OPT^H &= \text{the heavy part of $\OPT$}\\
	    \OPT_i &= \text{the bundle of agent $i$ in $\OPT$}\\
	    \OPT^H_i&= \text{the bundle of agent $i$ in $\OPT^H$}
	\end{align*}

\begin{definition}
	    The distance of two allocations is the number of edges that only exist in one of the allocations;
	    formally, the distance of two allocations $A$ and $B$ is $|A \oplus B|$.
	\end{definition}

\section{Properties of an Optimal Allocation}	
    
    In this section, we study the properties of an optimal allocation. The main property is stated in Theorem \ref{heavysim}. Roughly speaking, this theorem states that there exists an optimal allocation $\OPT$, in which heavy goods are assigned as evenly as possible. More formally, the utility profile of $\OPTH$ is leximax among all heavy-only allocations with the same cardinality. Later, we use this property to prove that the utility profile $\AH$ in the end of Algorithm \ref{alg2} is equal to the utility profile of $\OPT^H$, if $\OPT$ is chosen wisely among optimal allocations. After this, it will not be difficult to prove that the utility profiles of $A$ and $\OPT$ match.
    
    Let $\min(\OPT) = \min_i \uv_i(\OPT_i)$ be the minimum utility of any bundle in $\OPT$.
	
	\begin{lemma}
	    If $\uv_j(\OPT_j) \ge \min(\OPT) + 2$ then all goods in $\OPT_j$ are heavy for $j$. 
	\end{lemma} 
	\begin{proof}
	    Assume otherwise, and take a good that is light for $j$ and reallocate it to an agent $i$ for which $\uv_i(\OPT_i) = \min(\OPT)$. This will improve $\NSW$.
	\end{proof}
		
	\begin{corollary}\label{light-bundles}
		In $\OPT$ only bundles of utility $\min(\OPT)$ and $\min(\OPT) + 1$ can contain light goods. Bundles with higher value only contain goods that are heavy for the owner. 
	\end{corollary} 
	
	\begin{lemma}
	    There is no heavy alternating path starting with an $\OPT$-edge from agent $j$ to agent $i$ if $\uv_j(\OPT_j) > \uv_i(\OPT_i) + p$.  
	\end{lemma}
	\begin{proof}
	    Otherwise, augmentation of the path improves $\NSW$.
	\end{proof}
	
	\begin{lemma} \label{prop3}
	    If good $g$ is allocated as a light good to agent $i$, but could be allocated as a heavy good to agent $j$ who is allocated a light good $g'$, then the allocation is not optimal. 
	\end{lemma}
	\begin{proof}
	    Swapping the goods $g$ and $g'$ among agent $i$ and agent $j$, increases the value of agent $j$ by $p -1$ and the value of agent $i$ does not decrease.
	\end{proof}
	
	The rest of this section is dedicated to proving the following theorem. 
	
	\begin{theorem} \label{heavysim}
		Among all allocations with maximum $\NSW$, there exists an allocation $A$ such that the utility profile of $\AH$ is leximax among all heavy-only allocations of the same cardinality.
	\end{theorem}
	
	
	We choose $A$ and heavy-only $\CH$ as follows: (1) $A$ is an optimal allocation, (2) $\CH$ is leximax among all allocations of $\abs{\AH}$ heavy goods, and (3) the distance of $\AH$ and $\CH$ is minimum among all allocations satisfying (1) and (2). 
	
    Let us consider $\AH \oplus \CH$. 
    We label the edge with either $A$ or $C$ indicating whether it belongs to $\AH$ or $\CH$. 
    Note that in this graph, goods have degree zero, one, or two. We decompose the graph into edge-disjoint heavy alternating paths in the way that was described in section \ref{altpath}. We first show that there are no heavy alternating cycles.
    
    \begin{observation}\label{noCycle}
    There are no heavy alternating cycles in the decomposition.
    \end{observation}
    \begin{proof}
    Assume first that there is an alternating cycle, say $D$. Then $\CH \oplus D$ has the same utility profile as $\CH$ and is closer to $\AH$, a contradiction. 
    \end{proof}
    
    So we have only alternating paths. We now make more subtle observations about the edge-disjoint alternating paths in $\AH \oplus \CH$. 
    
    Alternating paths have either even or odd length. We next understand the even length alternating paths.
    Consider an even length alternating path, say $P$. The two endpoints of $P$ have the same kind, either both are goods or both are agents. First, we show that we cannot have an even length alternating paths with both endpoints as goods.
    
    \begin{observation}\label{noEvenLGoods}
    There are no even length heavy alternating paths with both endpoints as goods in the decomposition.
    \end{observation}
    \begin{proof}
    If both endpoints are goods, $\CH \oplus P$ has the same utility profile as $\CH$ and is closer to $\AH$, a contradiction. 
    \end{proof}
    
    Assume next that both endpoints are agents, say $i$ and $j$, and that the edge of $P$ incident to $i$ is in $\AH$ and the edge of $P$ incident to $j$ is in $\CH$. Then $\abs{\AH_i} > \abs{\CH_i}$ and $\abs{\CH_j} >\abs{\AH_j}$. 
    
    First, we show that $\abs{\AH_j} \leq \abs{\AH_i} -2$. Then using that, we prove that we can not have any even length alternating path with both endpoints as agents.
		
	\begin{observation}\label{jHasLessHeavyGoodThani}
	$\abs{\AH_j} < \abs{\AH_i}$.
	\end{observation}	
	\begin{proof}
	By contradiction.
		Assume first that 
		$\abs{\AH_j} \ge \abs{\AH_i}$.
		Then
		$\abs{\CH_j} > \abs{\CH_i}+1$
		and hence $\CH \oplus P$ is lexicographically larger than $\CH$, a contradiction.
	\end{proof}
		
	\begin{observation}\label{jHasLessHeavyGoodThani-1}
	$\abs{\AH_j} < \abs{\AH_i} -1$.
	\end{observation}
	\begin{proof}
	If 
		$\abs{\AH_j} = \abs{\AH_i} -1$,
		then $\AH\oplus P$ and $\AH$ have the same utility profile with respect to heavy goods. Also $\AH \oplus P$ is closer to $\CH$ than $\AH$. Finally, we swap the goods that are light for $i$ in $A_i$ with the goods that are light for $j$ in $A_j$. The value of the resulting bundle for $i$ or $j$  is at least the value of the other agent's bundle for the other agent in $A$. 
		Thus the resulting allocation is again optimal and with respect to heavy edges, it has the same utility profile as before and is closer to $\CH$, a contradiction.
	\end{proof}
		
	\begin{observation}\label{noEvenLAgents}
	There is no even length heavy alternating path with both endpoints as agents in the decomposition.
	\end{observation}	
	\begin{proof}
		By Observation \ref{jHasLessHeavyGoodThani-1}, we have $\abs{\AH_j} \le \abs{\AH_i} -2$.
		Consider $\AH \oplus P$. It is closer to $\CH$ than $\AH$. The value of bundle $A_j$ went up by $p$ and the value of bundle $A_i$ went down by $p$. If $A_j$ contains $p$ light goods for $j$, move them to $A_i$. In this way, we obtain an allocation which has at least the $\NSW$ of $A$ and where the allocation of heavy goods is closer to $\CH$. If $A_j$ contains less than $p$ light goods for $j$, then $\uv_i(A_i) > \uv_j(A_j) + p$ and $A_i$ contains no light goods for $i$ by Corollary \ref{altpath}. Figure \ref{fig5} shows the bundles of $i$ and $j$ before applying $\AH \oplus P$. Let $\ell$ be the number of light goods allocated to $j$ under $A$. We take these $\ell$ goods from $j$ and allocate them to $i$. In the resulting allocation $A'$ (which is shown in Figure \ref{fig6}), $\uv_i(A'_i) \geq \uv_i(A_i)-p+\ell$ and $\uv_j(A'_j)=\uv_j(A_j)+p-\ell$. Hence,
		\begin{align*}
		    \uv_i(A'_i)\uv_j(A'_j) &\geq (\uv_i(A_i)-p+\ell)(\uv_j(A_j)+p-\ell) \\
		    &= \uv_i(A_i)\uv_j(A_j)+(p-\ell)(\uv_i(A_i)-\uv_j(A_j) - (p-\ell))\\
		    &> \uv_i(A_i)\uv_j(A_j).
		\end{align*}
        The last inequality holds since $p > \ell$ and $\uv_i(A_i) > \uv_j(A_j) + p$. Since the bundles of other agents are not changed, $\NSW(A') > \NSW(A)$ which is a contradiction.	
	\end{proof}
	
	\begin{SCfigure} [10]
	\begin{tikzpicture}[
        node distance = 0mm,
            every node/.style = {draw, text height=1.5cm, inner sep=2mm, outer sep=0mm, align=center}]
            
            \node (i1) [fill=gray] {p};
            \node (others) [below=of i1, text height=0.1cm, fill=gray] {..};
            \node (i2) [below=of others, fill=gray] {p};
            \node (i3) [below=of i2, fill=gray] {p};
            \node (j1) [right=16mm of i3, fill=gray] {p};
            \node (j2) [above=of j1, text height=0.2cm] {1};
            \node (j3) [above=of j2, text height=0.2cm] {1};
    \end{tikzpicture}
    \caption{Bundles of agents $i$ and $j$ before applying $\AH \oplus P$.}
    \label{fig5}
	\end{SCfigure}
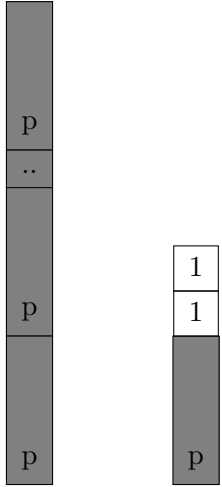	
	
	\begin{SCfigure}[10]
	\begin{tikzpicture}[
        node distance = 0mm,
            every node/.style = {draw, text height=1.5cm, inner sep=2mm, outer sep=0mm, align=center}]
            
            \node (others) [text height=0.1cm, fill=gray] {..};
            \node (i1) [below=of others, fill=gray] {p};
            \node (i2) [below=of i1, fill=gray] {p};
            \node (j2) [above=of others, text height=0.2cm] {1};
            \node (others-light) [above=of j2, text height=0.1cm] {..};
            \node (j3) [above=of others-light, text height=0.2cm] {1};
            \node (j1) [right=16mm of i2, fill=gray] {p};
            \node (i3) [above=of j1, fill=gray] {p};
    \end{tikzpicture}
    
    \caption{Bundles of agents $i$ and $j$ after applying $\AH \oplus P$ and moving the mentioned light goods. Note that in case any of these goods is not light for $i$, $u_i(A_i)$ will be even larger, resulting in an allocation with even better $\NSW$.}
    \label{fig6}
	\end{SCfigure}
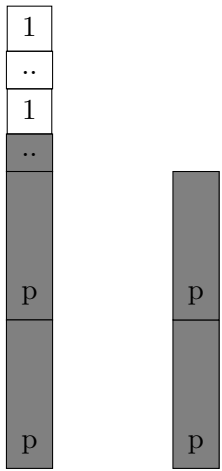	
	
	We now come to the case that we have no even length alternating path. Next we want to show that there is no odd length alternating path.
		
	\begin{observation}\label{noOddL}
	There is no odd length heavy alternating path in the decomposition.
	\end{observation}
	\begin{proof}
	Since $\abs{\AH} = \abs{\CH}$, if there is an odd length alternating path, there must be two odd length alternating path $P$ and $Q$, where $P$ starts and ends with an edge in $\AH$ and $Q$ starts and ends with an edge in $\CH$. The paths $P$ and $Q$ are edge-disjoint. For each path one of the endpoints is an agent and one is a good. Let $i$ be the agent endpoint of $P$ and $j$ be the agent endpoint of $P$. Then 
		$\abs{\AH_i} > \abs{\CH_i}$
		and 
		$\abs{\CH_j} > \abs{\AH_j}$.
 
		We now argue as above but with the even length alternating path replaced by $P \cup Q$. 
	\end{proof}
		

	\begin{proof}[Proof of Theorem \ref{heavysim}]
	Consider$\AH \oplus \CH$. By Observation \ref{noCycle}, we do not have any heavy alternating cycle. Also, by Observations \ref{noEvenLGoods} and \ref{noEvenLAgents}, there is no even length heavy alternating path. Moreover, by Observation \ref{noOddL}, there is no odd length heavy alternating path.
	Hence, we have $\AH = \CH$, and therefore, $\AH$ is leximax among all heavy-only allocation of the same cardinality.
	\end{proof}

	\begin{corollary}
	    Among all allocations maximizing $\NSW$, let $\OPT$ be such that $\OPTH$ is leximax among the heavy parts of allocations maximizing $\NSW$.
	    $\OPTH$ is leximax among all heavy-only allocations $\BH$ with $\abs{\BH} =\abs{\OPTH}$.
	\end{corollary}

	\newcommand{\LG}{\mathit{LG}}
	\begin{lemma} \label{greedy} Let $\OPT$ be an optimal allocation. Then the following allocation is also optimal. Start with $\OPTH$ and then allocate the goods in $\OPT^L$ greedily, i.e., allocate the goods one by one and for each good $g \in \OPT^L$ choose an arbitrary agent $i$ for which $\uv_i(\OPT'_i) = \min (\OPT')$, add edge $(i,g)$ to the current allocation $\OPT'$, and update $\OPT'$.  
	\end{lemma}
	
	\begin{proof}
	Consider a sequence of assigning the goods in $\OPT^L$ which results in $\OPT$ and is closest to greedily assigning a good to an agent with minimum utility. Assume that at some point with partial allocation $X$, we assign a good $g$ to an agent $i$ but there is another agent $j$ with minimum utility such that $u_j(X_j)<u_i(X_i)$. Note that $u_i(X_i \cup \{g\}) > u_j(X_j)+1$ and $i$ has a light good. By Corollary \ref{light-bundles}, $u_i(\OPT_i) - u_j(\OPT_j) \leq 1$. Hence, after assigning $g$ to $i$, a light good $h$ should be assigned to $j$ as well. First assigning $h$ to $j$ and then $g$ to $i$ makes the sequence closer to a greedy sequence which is a contradiction. Therefore, $u_i(X_i)$ is minimum and the sequence is in fact greedily assigning goods to agents with minimum utility.
	\end{proof}
	
    Before we go on to explain the algorithm, it is worth mentioning why our approach does not work when $p$ is not an integer. Theorem \ref{heavysim} is not true if $p$ is half-integer, say $p = 3/2$. Consider an instance with two agents, two goods that are heavy for both agents and three goods that are light for both agents. In the optimal allocation, both agents have bundles of value 3. The bundle of one agent contains the two goods that are heavy for her and the bundle of the other agents contains three goods that are all light for her. The heavy part of this allocation is not leximax among all allocations in which two goods are allocated as heavy goods.

\section{Algorithm}
    In this section we elaborate the algorithm. Our algorithm operates in $3$ phases. The first phase finds a heavy-only allocation which maximizes the $\NSW$. This phase is equivalent to maximizing $\NSW$ in a binary instance. Barman et al.\ \cite{barman2018binarynsw} proved that this is possible in polynomial time.
    
    In the second phase, we greedily allocate the remaining goods (one by one) to an agent with minimum utility. Note that all these goods are light for all agents. Otherwise, the output of the first phase does not maximize $\NSW$ among all heavy-only allocations. This is why this phase is called ``allocating light goods''.
    
    The third phase is the most technical one in which we reshuffle some of the goods. More precisely, we take a heavy good from the bundle of an agent with maximum utility and allocate it to an agent with minimum utility as long as $\NSW$ increases. We later show in Lemma \ref{importantLemma} that the reallocated goods are light for their new owners. This means, as long as there is a progress, we turn some of the heavy goods into light goods.
    
    \paragraph{Phase 1:} Heavy-Only Allocations
    
    As a first step we concentrate on heavy-only allocations. We first show how to compute a heavy-only allocation maximizing $\NSW$. We use our own words to describe Algorithm $\ref{alg1}$ in \cite{barman2018binarynsw}. 
	
	In order to compute a heavy-only allocation maximizing $\NSW$, we start with a heavy-only allocation $A$ of maximum cardinality, i.e., in $A$ any good that is heavy for at least one agent is assigned to an agent for which it is heavy. We then improve the $\NSW$ of $A$ by augmentation of some heavy alternating paths. As long as there is a heavy even-length alternating path $P$ connecting agent $i$ to agent $j$, starting with an edge outside $A$ and ending with an edge in $A$, and with the heavy degree of $j$ at least two larger than the heavy degree of $i$ in $A$, we augment $P$ to $A$, i.e., we update $A$ to $A \oplus P$. When the process stops, $A$ maximizes $\NSW$. 

	\begin{algorithm} 
		\caption{BinaryMaxNSW} \label{alg1}
		Input : $N, M, v = (\uv_1, ... , \uv_n)$ 
		
		Output: allocation $A$
		
		\begin{algorithmic}
			\State let $G$ be the corresponding graph
			\State find the maximum multi-matching $A$
			\While {there is an alternating path $a_0, g_1, ... , g_k, a_k$ s.t $|A_0| \le |A_k| - 2$}
			\For {$\ell \leftarrow k$ to $1$}
			\State $A \leftarrow A \backslash (a_{\ell}, g_{\ell})$
			\State $A \leftarrow A \cup (a_{\ell-1}, g_{\ell})$
			\EndFor		
			\EndWhile
			return $A$
		\end{algorithmic}
		
	\end{algorithm}
	Barman et al.\ showed in \cite{barman2018binarynsw} that Algorithm~\ref{alg1} outputs an allocation with maximum $\NSW$. Furthermore, Halpern et al.\ proved in \cite{HalpernPPS20} that in binary instances the set of leximax allocations is identical to the set of allocations with maximum $\NSW$.
	
	\begin{theorem} \label{binaryNSW}
		Having an instance with only heavy edges as an input, Algorithm \ref{alg1} outputs an allocation with maximum Nash Social Welfare. Furthermore, an optimal allocation $\OPT$ is leximax and hence the utility profile of the optimal allocation is unique. 
	\end{theorem}

	Before proceeding to the next phase, we briefly explain how to get leximax heavy-only (partial) allocations of different cardinalities. The heavy-only allocation maximizing $\NSW$ is a maximum cardinality heavy-only allocation. In order to compute heavy-only allocations of smaller cardinality, we repeatedly remove an edge from $A$. We take any bundle $A_i$ with $\uv_i(A_i) = \max_j \uv_j(A_j)$ and remove an edge of $A$ incident to $i$. In this way, we will obtain optimal allocations for every cardinality. 
	
	\begin{lemma}\label{lemmaOne}
	    Let $\CH \subseteq \EH$ be leximax among all allocations $D^H \subseteq \EH$ with $|D^H|= |\CH|$. Let $(c_1, c_2, \dots , c_n)$ be the utility profile of $\CH$ and let $t$ be such that $c_{t-1} < c_t = \dots = c_n$. Then allocation $\hat{C}^H$ with utility profile $(\hat{c}_1 , \dots , \hat{c}_n) = (c_1, \dots, c_{t-1}, c_t - p, c_{t+1}, \dots , c_n)$, is leximax among all heavy-only allocations with the same cardinality. 
	\end{lemma}
	 
	\begin{proof}
	    Let $\hat{D}^H$ be leximax among all allocations with $|\hat{C}^H|$ many heavy allocated goods. Let $(\hat{d}_1, \dots, \hat{d}_n)$ be the utility profile of $\hat{D}^H$. Note that since $|\hat{D}^H| = |\hat{C}^H| < |\CH|$, there exists a good $g$ that is unallocated under $\hat{D}^H$ and is of value $p$ for some agent $a$.
		
		Consider the smallest $i$ such that $\hat{d}_i \not= \hat{c}_i$. Then $\hat{d}_i > \hat{c}_i$ since $\hat{D}^H$ is leximax. If $i < t$, allocating $g$ to agent $a$ results in an allocation $D^H$ with $|D^H| = |\CH|$ which is lexicographically larger than $\CH$. This contradicts the choice of $\CH$.
		
		So  $ \hat{d}_i = \hat{c}_i$ for all $i < t$ and hence
		$ \Sigma_{i = t}^{n} \hat{d}_i = \Sigma_{i = t}^{n} \hat{c}_i = (n-t+1)c_n - 1$. Since  $(c_n - 1, c_n, \dots , c_n)$ is leximax among all $(n-t+1)$ tuples with sum $(n-t+1)c_n -1$,
		$ (\hat{d}_t , \dots , \hat{d}_n) \preceq_{lex} (c_n - 1, \dots , c_n) = (\hat{c}_{t} , \dots , \hat{c}_n)$. Hence $\hat{C}^H$ is leximax among all allocations $\hat{D}^H$ with $|\hat{D}^H| = |\hat{C}^H|$.  
	\end{proof}
	
	\begin{corollary} \label{subset}
	    Let $(p \cdot a_1, ... , p \cdot a_n)$ and $(p \cdot b_1, ... , p \cdot b_n)$ be the utility profile of heavy-only allocations $A$ and $B$. Note that $\Sigma_{i=1}^n a_i = |\AH|$ and $ \Sigma_{i=1}^n b_i = |\BH|$ and let $|\AH| \leq |\BH|$. If the utility profile of $A$ is leximax among all the utility profiles of heavy-only allocations $C$ with $|C| = |A|$ and same holds for $B$, then for all $1 \leq i \leq n: a_i \leq b_i$ .
	\end{corollary}
		
	\begin{proof}
		Keep removing goods from the bundle with maximum utility and minimum index in $B$ until we reach an allocation $\hat{B}$ with $|\hat{B}^H| = |\AH|$. By Lemma \ref{lemmaOne}, $(p \cdot \hat{b}_1, ..., p \cdot \hat{b}_n) \succeq_{lex} (p \cdot a_1, ... , p \cdot a_n)$ and therefore $(\hat{b}_1, ..., \hat{b}_n) = (a_1, ... , a_n)$. The fact that $\hat{b}_i \leq b_i$ for all $i \in [n]$, completes the proof.
	\end{proof}
	
	
	\paragraph{Phase 2:} Allocating Light Goods 
	
	As long as there is an unallocated good, allocate it to an agent with minimum utility.
	
	\paragraph{Phase 3:}  Increasing $\NSW$
	
	As long as $\NSW$ increases, take a good from an agent with maximum utility and give it to an agent with minimum utility. See Algorithm \ref{alg2}.

	\begin{algorithm}
		\caption{TwoValueMaxNSW} \label{alg2}
		Input : $N, M, \uv = (\uv_1, ... , \uv_n)$ 
		
		Output: allocation $A$
		\begin{algorithmic}[1]
			\State  $\backslash\backslash$ Phase 1
			\State let $\uv'_i: 2^M \rightarrow \mathbb{N}$ be an additive function and $\uv'_i(g) = \floor{\uv_i(g) / p}$ for all agents $i$ and all goods $g$
			\State $\uv' \leftarrow (\uv'_1, \dots , \uv'_n)$
			\State $A = $BinaryMaxNSW$(N, M, u')$
			
			\State  $\backslash\backslash$ Phase 2
			\While {there is an unallocated good $g$}
			\State let $\uv_1(A_1) \le \uv_2(A_2) \le ... \le \uv_n(A_n)$
			\State let $k$ be the maximum index s.t $\uv_k(A_k) = \uv_1(A_1)$
			\State $A \leftarrow A \cup \{(k,g)\}$
			\EndWhile		
			
			\State  $\backslash\backslash$ Phase 3
			
			\While{$ \uv_n(A_n) >  p \cdot \uv_1(A_1) + p $}
			\State let $k$ be the maximum index s.t $\uv_k(A_k) = \uv_1(A_1)$
			\State let $t$ be the minimum index s.t $\uv_t(A_t) = \uv_n(A_n)$
			\State let $g$ be a good such that $(t,g) \in A$
			\State $A \leftarrow A \backslash \{(t,g)\}$
			\State $A \leftarrow A \cup \{(k,g)\}$   
			\EndWhile
			\State return $A$
		\end{algorithmic}
	\end{algorithm}

	\section{Correctness}

    Phase 1 already gives us an optimal allocation $X$ of the heavy only goods which is also leximax on the allocation of the heavy only goods. In phase 2, we allocate the small valued goods as ``evenly'' as possible. The only reason why our solution may not be optimal is that the number of heavy goods in an optimal allocation $Y$ may be less than that in $X$. However, if this is the case, then we can move from $X$ to $Y$ by making small local improvements in Nash social welfare by moving heavy goods from one bundle to the other (captured by Theorem~\ref{oneCase} and Corollary~\ref{meetingOPT}).
    
	Let $A$ be the allocation that Algorithm \ref{alg2} outputs. First we prove that there is an allocation $\OPT$ with maximum $\NSW$ such that the utility profile of $\OPT^H$ and $\AH$ are the same. Then we prove that in allocation $\OPT$, the remaining goods are allocated the same way as in $A$. 
	
	We start with showing some invariants of Algorithm \ref{alg2}.

 \begin{lemma}\label{importantLemma} Fix a numbering of the agents at the beginning of phase 3 such that $\uv_1(A_1) \le \uv_2(A_2) \le \ldots \le \uv_{n-1}(A_{n-1}) \le \uv_n(A_n)$. During phase 3, the following holds. 
  \renewcommand{\theenumi}{\textrm{\alph{enumi}}}

  \begin{enumerate}
    \item \label{neatOrder} The ordering $\uv_1(A_1) \le \uv_2(A_2) \le \ldots \le \uv_{n-1}(A_{n-1}) \le \uv_n(A_n)$ is maintained.
    \item \label{oneStep} If $A_i$ contains a good that is light for $i$, then $\uv_i(A_i) \le \uv_1(A_1) + 1$. 
    \item \label{lexmax} $\AH$ is leximax among all heavy-only allocations of the same cardinality. 
    \item Whenever a good is moved in phase 3, say from bundle $A_t$ to bundle $A_k$, all goods in $A_t$ are heavy for $t$ and light for $k$.
  \end{enumerate}
\end{lemma}
\begin{proof} We prove statements a) to d) by induction on the number of iterations in phase 3. Before the first iteration a) and d) trivially hold. Claim b) holds since in phase 2 we allocate only goods that are light for every agent and since the next good is always added to a lightest bundle. Claim c) holds by Theorem \ref{binaryNSW}.

  Assume now that a) to c) hold before the $i$-th iteration and that we move a good $g$ from $A_t$ to $A_k$ in iteration $i$. We will show that d) holds for $A_t$ and $A_k$ and that a) to c) hold after iteration $i$.

  By the condition of the while-loop, we have $\uv_t(A_t) > p \cdot (\uv_k(A_k) + 1)$. Thus $A_t$ contains only goods that are heavy for $t$ by part b) of the induction hypothesis. Let $g$ be any good in $A_t$. If we also have $\uv_k(g) = p$, then moving $g$ from $A_t$ to $A_k$ would result in an allocation of heavy goods that is lexicographically larger, a contradiction to c). Thus $\uv_k(g) = 1$.

  After moving $g$, c) holds by lemma \ref{lemmaOne}. Note that $g$ is given from an agent with maximum utility and is not heavy for its new owner. 
  
  Since $k$ is the largest index such that $\uv_k(A_k) = \uv_1(A_1)$ before the $i$-th iteration, b) holds after the $i$-th iteration.

  It remains to show that part a) holds after the $i$-th iteration. The weight of the $k$-th bundle increases by 1 and the weight of the $t$-bundle decreases by $p$. We need to show $\uv_t(A_t) \ge \uv_{t-1}(A_{t-1}) +p +\delta$, where $\delta =1 $ if $k = t-1$ and $\delta = 0$ otherwise. 
  \begin{itemize}
    \item If $k = t-1$, we have 
    $\uv_t(A_t) \ge p \cdot (\uv_{t-1}(A_{t-1} )+ 1) + 1$ and hence $\uv_t(A_t) - \uv_{t-1}(A_{t-1}) - p - 1 \ge (p-1) \cdot \uv_{t-1}(A_{t-1})  \ge 0$. 
    \item If $k < t-1$, by definition of $t$, $\uv_t(A_t) > \uv_{t-1}(A_{t-1})$. If all goods in $A_{t-1}$ are heavy for $t-1$, the difference in weight is at least $p$ and we are done. 
    If $A_{t-1}$ contains a good that is light for $t-1$, then $\uv_{t-1}(A_{t-1}) \le \uv_k(A_k) + 1$ by condition b) and hence $\uv_t(A_t) \ge p \cdot \uv_{t-1}(A_{t-1}) + 1$. This implies $\uv_t(A_t) \ge \uv_{t-1}(A_{t-1}) + p$ except if $\uv_{t-1}(A_{t-1}) = 0$. In the latter case, $k = t-1$, a case we have already dealt with.  
  \end{itemize}    
  
  We also need to show that after moving the good, $\uv_k(A_k) \leq \uv_{k+1}(A_{k+1})$. By the choice of $k$ and the fact that $u_i(X_i) \in \mathbb{N}$, $\uv_k(A_k) \leq \uv_{k+1}(A_{k+1}) + 1$ holds before moving the good. After moving the good, by condition d), $\uv_k(A_k)$ increases by $1$ and therefore, $\uv_k(A_k) \leq \uv_{k+1}(A_{k+1})$.
  \end{proof}

	\begin{theorem}\label{oneCase}
		Let $\OPT$ be an allocation that maximizes $\NSW$ and subject to that, maximizes $|\OPTH|$. Let $A$ be the output of algorithm \ref{alg2}. Then $|\OPTH| \ge |\AH|$. 
	\end{theorem}

	\begin{proof} 
		Assume $|\OPTH| < |\AH|$. By the choice of $\OPT$, $A$ cannot maximize $\NSW$. We first show that we may assume $\abs{\OPTH_i} \le \abs{\AH_i}$ for all $i$. We can obtain a leximax heavy-only allocation $\CH$ of cardinality $\abs{\OPTH}$ from $\AH$ by repeatedly removing a good from the lowest indexed bundle of maximum utility. The utility profiles of $\CH$ and $\OPTH$ agree and hence there is a bijection $\pi$ of the set of agents such that $\abs{\CH_i} = \abs{\OPTH_{\pi(i)}}$. 
		Let $\ell_i$ be the number of goods in $\OPTH_{\pi(i)}$ which are light for $\pi(i)$. Note that the number of goods which are not allocated under $C^H$ is equal to the number of light goods under $\OPT$, i.e, $\Sigma_{i \in [n]} \ell_i$. 
		Obtain an allocation $C$ from $\CH$ by giving $\ell_i$ not yet allocated goods to $C_i$. Then $\uv_i(C_i) \ge \uv_{\pi(i)}(OPT_{\pi(i)})$ for all $i$. Thus, $C$ is optimal and $\uv_i(C_i) = \uv_{\pi(i)}(OPT_{\pi(i)})$. Also $\abs{C_i} \le \abs{\AH_i}$ for all $i$. We may therefore assume $\abs{\OPTH_i} \le \abs{\AH_i}$ for all $i$.
		
		Since $C$ is optimal, Lemma \ref{greedy} gives us an alternative way of obtaining an optimal allocation. Start from $\CH$ and then allocate the goods that are allocated as light goods (i.e. the goods that are light for their owner) in $\OPT$ in a greedy fashion. Let then $R$ be the set of agents from which we removed a good in moving from $\AH$ to $\CH$. Note that no good is added to the bundle of an agent in $R$ when adding goods greedily. Otherwise, we would have a contradiction to Lemma~\ref{prop3}.

	Since $A$ is not optimal there must be an agent $k$ such that $\uv_k(\OPT_k) > \uv_k(A_k)$. The bundle $\OPT_k$ must contain a good $g$ that is light for $k$. 
		
		Since $A$ is the output of Algorithm \ref{alg2}, moving a good from $A_n$ to $A_1$ does not increase $\NSW$. So
		\begin{align*}
		(\uv_1(A_1) + 1)(\uv_n(A_n - p)) &\leq \uv_1(A_1)\uv_n(A_n) \\
		\intertext{and hence,} \uv_n(A_n) &\leq p \cdot \uv_1(A_1) + p.
		\end{align*}
		
		Consider $\OPTH \oplus \AH$ and its alternating path decomposition. Since the number of heavy edges in $\OPT$ is less than the number of heavy edges in $A$, there must be alternating path $P$ starting with an edge in $A$ and ending in an edge in $A$. Let $t$ and $g$ be the endpoints of the path; $t$ is an agent and $g$ is a good. Then $g$ must be allocated in $\OPT$ as a light good to an agent $j$ since $g$ has degree one in $\OPTH \oplus \AH$. 

        Since $t \in R$, $\abs{\OPTH_t} < \abs{\AH_t}$, and no good is added to $\OPT_t$ in the greedy assignment of goods,	$\uv_t(\OPT_t) + p \leq \uv_t(A_t) \leq \uv_n(A_n)$. We also have $\uv_k(\OPT_k) > \uv_k(A_k) \geq \uv_1(A_1)$. Therefore, we get	

		\begin{align*} \label{ineq}
		\uv_t(\OPT_t) + p \le \uv_n(A_n) &\le p \cdot \uv_1(A_1) + p \le p \cdot \uv_k(\OPT_k)
		\end{align*}
		and hence,
		\begin{align*}
		\uv_t(\OPT_t) \leq p\cdot \uv_k(\OPT_k) - p.
		\end{align*}
		
		Taking $g$ from $j$'s bundle and changing $\OPT$ to $\OPT \oplus P$ increases the number of heavy goods allocated to $t$ by one. In case $k \neq j$, we replenish $j$ from $\OPT_k$ by taking a light good from $\OPT_k$ and allocating it to $\OPT_j$. This reallocation of goods does not decrease $\NSW$ as:
		\begin{align}
		(\uv_k(\OPT_k) - 1)(\uv_t(\OPT_t) + p) - \uv_k(\OPT_k)\uv_t(\OPT_t) = p \cdot \uv_k(\OPT_k) - \uv_t(\OPT_t) - p \ge 0.
		\end{align}	
		
		For the new allocation $\widehat{\OPT}$ we have $\NSW (\widehat{\OPT}) \geq \NSW (\OPT)$ and $|\widehat{\OPT}^H| > |\OPTH|$ which contradicts the choice of $\OPT$.
	\end{proof}
	
	\begin{lemma}\label{thisLemma}
		Let $\hat{A}$ be the partial allocation after phase $1$ of the Algorithm \ref{alg2}. Then $|\hat{A}^H| \ge |\OPTH|$ for any optimal allocation $\OPT$. 
	\end{lemma}
	
	\begin{proof}
		Assume otherwise. Then there should be a heavy edge $(i,g)$ which is not in $\hat{A}$. After allocating $g$ to agent $i$, $\NSW$ increases. This contradicts Theorem \ref{binaryNSW}.
	\end{proof}
	
	By Lemma \ref{thisLemma} and Theorem \ref{oneCase}, we can assume in some round in phase $3$ of Algorithm \ref{alg2} with allocation $\Tilde{A}$, $|\Tilde{A}^H| = |\OPTH|$. Then, by Lemma \ref{importantLemma}.\ref{lexmax} and Theorem \ref{heavysim}, we get the following Corollary.
	
	\begin{corollary}\label{meetingOPT}
	    There is an optimal allocation $\OPT$ such that $\Tilde{A}^H = \OPTH$.
	\end{corollary}
	
	So far, we have proved that considering only heavy allocated goods in $\Tilde{A}$ and $\OPT$, we end up having the same utility profile. By Lemma \ref{importantLemma}.\ref{neatOrder} and Lemma \ref{greedy}, in both allocations $\OPT$ and $\Tilde{A}$, light goods are allocated as evenly as possible.

	So we can conclude that the utility profiles of $\Tilde{A}$ and $\OPT$ are equal. In each round of the phase $3$ of Algorithm \ref{alg2}, Nash Social Welfare increases. This means $\NSW (A) \ge \NSW (\Tilde{A}) = \NSW (\OPT)$. 

	\begin{theorem}\label{IntegerResult}
	There exists a polynomial-time algorithm for the Nash Social Welfare problem with 2-value instances $1$ and $p$ when $p$ is an integer.
	\end{theorem}
	
	\begin{proof}
	    We already proved that the output of Algorithm \ref{alg2} is a Nash social welfare maximizing allocation. It only remains to prove that this algorithm is polynomial-time. By \cite{barman2018binarynsw}, Algorithm \ref{alg1} and hence the first phase of Algorithm \ref{alg2} runs in polynomial time. The second phase clearly takes polynomial time. By lemma \ref{importantLemma}.d, the number of heavy goods under $A$ is decreasing after each iteration of the third phase. Therefore, this phase can be run at most $m$ times. All in all, Algorithm \ref{alg2} terminates in polynomial time. 
	\end{proof}
	
	\section{Approximation for General Two Value Instances}
	
	In this section, we want to introduce an approximation algorithm for general 2-value instances. Therefore we can assume that the values are $1$ and any real value $p>1$. The idea is to round $p$ to $\floor{p}$ or $\ceil{p}$ to have an integer value and then run Algorithm \ref{alg2}. We will show that this results in an approximation factor of $\max \{ \frac{\floor{p}}{p}, \frac{p}{\ceil{p}}\}$ for $\NSW$. Since $\max \{ \frac{\floor{p}}{p}, \frac{p}{\ceil{p}}\} \geq \sqrt{2}$, we improve the state-of-the-art $e^{1/e}$ approximation of the maximum Nash social welfare in 2-value instances.

	\begin{algorithm}
		\caption{ApproximationTwoValueMaxNSW} \label{alg3}
		Input : $N, M, \uv = (\uv_1, ... , \uv_n)$ 
		
		Output: allocation $\APXAlg$
		\begin{algorithmic}[1]
			\If{$\frac{\floor{p}}{p} \geq \frac{p}{\ceil{p}}$}:
			
			\State let $u'_i: 2^M \rightarrow \mathbb{N}$ be an additive function and $u'_i(g) = \floor{u_i(g)}$ for all agents $i$
			
			\Else:
			\State let $u'_i: 2^M \rightarrow \mathbb{N}$ be an additive function and $u'_i(g) = \ceil{u_i(g)}$ for all agents $i$
			\EndIf
            \State $\uv' \leftarrow (\uv'_1, \dots , \uv'_n)$
			\State $\APXAlg = $TwoValueMaxNSW$(N, M, u')$
			\State return $\APXAlg$
		\end{algorithmic}
	\end{algorithm}
	
	\begin{theorem}\label{APXResult}
	There exists a polynomial-time $\max \{ \frac{\floor{p}}{p}, \frac{p}{\ceil{p}}\}-$approximation algorithm for the Nash Social Welfare problem with 2-value instances $1$ and $p>1$.
	\end{theorem}

	\begin{proof}
	First assume $\frac{\floor{p}}{p} \geq \frac{p}{\ceil{p}}$.
	Consider $\OPT$ which is the optimum allocation under utility vector $\uv$. Consider this allocation under utility vector $\uv'$. The utility of the goods are either not changed, or are multiplied by $\floor{p} / p$. Therefore,
	\begin{align*}
	   \NSW(\OPT,\uv') \geq \NSW(\OPT,\uv) \cdot \frac{\floor{p}}{p}.
	\end{align*}	
	Since $p \geq \floor{p}$ and Algorithm \ref{alg2} gives the optimum allocation under utility vector $\uv'$, we have
	\begin{align*}
	    \NSW (\APXAlg,\uv) \geq \NSW (\APXAlg, \uv') \geq \NSW (\OPT,\uv').
	\end{align*}
	Therefore we can conclude,
	\begin{align*}
	    \NSW(\APXAlg, \uv) \geq \NSW(\OPT, \uv) \cdot \frac{\floor{p}}{p}.
	\end{align*}
	Now consider the second case in which $\frac{\floor{p}}{p} \leq \frac{p}{\ceil{p}}$. 
	Consider allocation $\APXAlg$ under utility vector $\uv$. With comparison to the utility vector $\uv'$, the utility of the goods are either not changed, or are multiplied by $ \ceil{p}/p$. Therefore,
	\begin{align*}
	    \NSW(\APXAlg, \uv) \geq \NSW (\APXAlg, \uv') \cdot \frac{p}{\ceil{p}}.
	\end{align*}
	Now consider $\OPT$ which is the optimum allocation under utility vector $\uv$. Since Algorithm \ref{alg2} gives the optimum allocation under $\uv'$ and $\ceil{p} \geq p$, we have
	\begin{align*}
	   \NSW (\APXAlg, \uv') \geq \NSW(\OPT, \uv') \geq \NSW(\OPT, \uv).
	\end{align*}
	Therefore we can conclude,
	\begin{align*}
	    \NSW(\APXAlg, \uv) \geq \NSW (\OPT, \uv) \cdot \frac{p}{\ceil{p}} .
	\end{align*}
	\end{proof}